\documentclass[sigconf, nonacm]{acmart}

\usepackage{subcaption}
\usepackage{xspace}
\usepackage{tabularx}
\usepackage{algorithm}
\usepackage{algpseudocode}
\usepackage{amsmath}
\usepackage{optidef}

\usepackage{amsthm}
\theoremstyle{plain}
\newtheorem{theorem}{Theorem}

\newtheorem{corollary}{Corollary}[theorem]
\theoremstyle{definition}

\theoremstyle{remark}
\newtheorem{remark}{Remark}

\newtheorem{proposition}{Proposition}[theorem]

\newcommand{\framework}{\textsc{ListK}\xspace}
\newcommand{\qsort}{\textsc{LMPQSort}\xspace}
\newcommand{\qselect}{\textsc{LMPQSelect}\xspace}
\newcommand{\lmpq}{\textsc{LMPQ}\xspace}
\newcommand{\tfilter}{\textsc{LTFilter}\xspace}
\newcommand{\tselect}{\textsc{LTTopK}\xspace}
\newcommand{\lotustopk}{\textsc{LOTUSTopK}\xspace}
\newcommand{\pairwise}{\textsc{Pairwise}\xspace}
\newcommand{\pointwise}{\textsc{Pointwise}\xspace}

\newcommand{\thh}{\textsuperscript{th}\xspace}

\newcommand{\sect}[1]{Section~\ref{#1}}
\newcommand{\tbl}[1]{Table~\ref{#1}}
\newcommand{\fig}[1]{Figure~\ref{#1}}
\newcommand{\alg}[1]{Algorithm~\ref{#1}}

\newcommand{\cor}[1]{Corollary~\ref{#1}}
\newcommand{\thm}[1]{Theorem~\ref{#1}}
\newcommand{\eq}[1]{Equation~\ref{#1}}

\newcommand{\fn}[1]{\textcolor{black}{#1}}

\newcommand{\fixed}[1]{\textcolor{black}{#1}}

\begin{document}
\title{\textsc{ListK}: Semantic ORDER BY and LIMIT K with Listwise Prompting}

%%
%% The "author" command and its associated commands are used to define the authors and their affiliations.
\author{Jason Shin}
\authornote{Both authors contributed equally.}
\affiliation{%
  \institution{University of Rochester}
  \city{Rochester}
  \state{New York}
  \country{USA}
}
\email{jshin60@u.rohcester.edu}

\author{Jiwon Chang}
\authornotemark[1]
\affiliation{%
  \institution{University of Rochester}
  \city{Rochester}
  \state{New York}
  \country{USA}
}
\email{jchang38@ur.rochester.edu}

\author{Fatemeh Nargesian}
\orcid{0000-0001-5109-3700}
\affiliation{%
  \institution{University of Rochester}
  \city{Rochester}
  \state{New York}
  \country{USA}
}
\email{fnargesian@rochester.edu}


\begin{abstract}
Semantic operators abstract large language model (LLM) calls in SQL clauses. 
It is gaining traction as an easy method to analyze semi-structured, unstructured, and multimodal datasets. 
While a plethora of recent works optimize various semantic operators, existing methods for semantic ORDER BY (full sort) and LIMIT K (top-K) remain lackluster. 
%Semantic sort and top-K are useful when the user's intent is a ranking function expressed with natural language in terms of relevance or some other discrimination criteria. 
Our \framework framework improves the latency of semantic ORDER BY ... LIMIT K at no cost to accuracy.  
Motivated by the recent advance in fine-tuned listwise rankers, we study several sorting algorithms that best combine partial listwise rankings. 
These include: 1) deterministic listwise tournament (\tselect), 2) Las Vegas and embarrassingly parallel listwise multi-pivot quickselect/sort (\qselect, \qsort), and 3) a basic Monte Carlo listwise tournament filter (\tfilter). 
Of these, listwise multi-pivot quickselect/sort are studied here for the first time. 
The full framework provides a query optimizer for combining the above physical operators based on the target recall to minimize latency. 
We provide theoretical analysis to easily tune parameters and provide cost estimates for query optimizers. 
\framework empirically dominates the Pareto frontier, halving latency at virtually no cost to recall and NDCG compared to prior art. 
\end{abstract}

\maketitle

\begin{figure*}[!t]
    \centering
    \includegraphics[width=\linewidth]{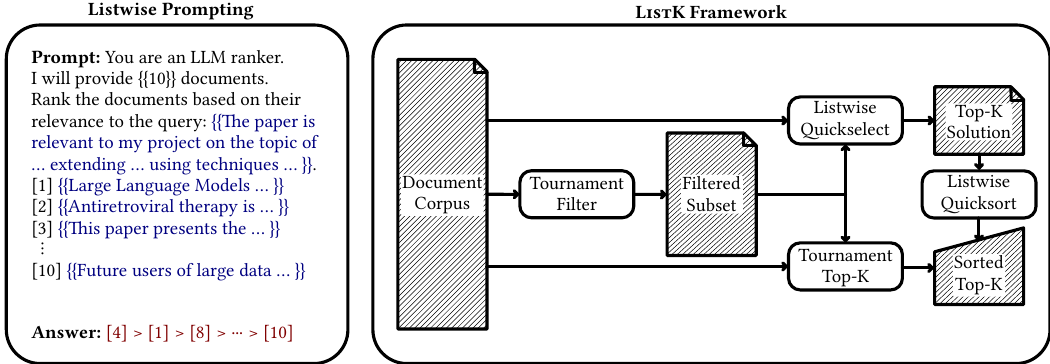}
    \caption{A visualization of our overall strategy. (Left) An example of how listwise prompting is used to extract rankings from an LLM. (Right) Four physical plans supported by \framework: \textsc{LMPQSelect/Sort} and \tselect with optional \tfilter.}
    \label{fig:framework}
    
    \Description{Left box shows a listwise prompt that shows the documents to the LLM, and the LLM answers with a full ranking. Right box shows a diagram of data flow from document corpus to an optional filtered subset through the tournament filter. Then, the full corpus or a filtered subset flows through either tournament top-K to yield sorted top-K, or through listwise quickselect to yield an unordered top-K which is then sorted by listwise quicksort.}
    
\end{figure*}

\begin{table*}[t!]
    \centering
    \begin{tabularx}{\linewidth}{|c|c|c|c|X|}
    \hline
    Algorithm & Purpose & (Non-)Determinism & Prompting  & \# of LLM calls \\
    \hline
    \tselect & select+sort & deterministic & listwise & $\displaystyle \frac{1}{L-1} \cdot (N-1) + (K-1) \cdot \log(N/L)$ \hfill (\sect{sec:theory:result:tourney}) \\
    \qselect & selection & Las Vegas & listwise & $\displaystyle \frac{(P + 1)}{(L - P) (P - 1 + (K/N)^{P+1} + (1-K/N)^{P+1}} \cdot N + O(1)$ \hfill (\thm{thm:quickselect-expectation}) \\
    \qsort & sorting & Las Vegas & listwise & $\displaystyle \frac{1}{(L-P) \cdot \log (P+1)} \cdot N \log (N/L) + O(N)$ \hfill (\thm{thm:quicksort-expectation}) \\
    \tfilter & pre-filter & Monte Carlo & listwise & $N/L$ \hfill(\sect{sec:theory:result:filter})  \\
    \lotustopk & select+sort & Las Vegas & pairwise & $2.5N + K \log K$ \hfill
    \cite{lotus,quickselect,quicksort} \\
    \hline
    \end{tabularx}
    \caption{Summary of algorithms that we study compared to prior art.}
    \label{tab:algorithms-comparison}
\end{table*}

\section{Introduction} \label{sec:intro}

\subsection{Semantic ORDER BY and Top-K}
\fn{With the recent proliferation of large language models (LLMs), it is now possible to express query objectives using natural language. This allows us to perform data analytics on semi-structured or unstructured documents and multimodal data.}

\paragraph{Semantic Operators} 
While there are many different ways of integrating LLMs into data analytics systems, one promising method is semantic operators (SemOps)~\cite{semop}. 
Semantic operators parameterize various SQL clauses with natural language \fn{to improve expressivity.}  
SemOps are a large improvement over User-defined Functions (UDFs). 
SemOps are high-level ``transparent'' logical operators, because the natural language clauses are embedded in the query and therefore visible to the DBMS. 
On the other hand, LLM UDFs are ``opaque'' and written in an imperative language, with calls to external models. 
As such, these UDFs admit a more limited scope of optimizations. 
Academic and industrial implementations of SemOps support a variety of logical and physical operators~\cite{lotus,thalamusdb,palimpzest,bigquery,uqe,databricks_ai_query}. 
In this paper, we focus on the semantic top-K operator. 

\paragraph{Semantic Top-K Operator} Semantic top-K queries are critical in any application where there are more items than a human can process. 
It encompasses finding the top items ($K = 1, 2, \ldots$) as well as ORDER BY query for full sort. 
Henceforth, when we refer to top-K, it is a shorthand that refers to both scenarios. 
We give two examples below. 
\begin{verbatim}
    SELECT * FROM science_papers ORDER BY "The paper 
    is relevant to my project on the topic of ... 
    extending prior works by ... using techniques 
    like ...." LIMIT 20;
\end{verbatim}
\begin{verbatim}
    SELECT \* FROM resumes ORDER BY "The candidate is the
    top selection with our open role ... which has the 
    following job description: ....";
\end{verbatim}
In both cases, the queries are ranking and returning unstructured documents by their relevance to complex user priorities that can only be accurately described using natural language. 
Hence, the in-database support of semantic top-K operators is the ideal tool for these tasks. 
While these queries can generally be applied to items of structured or unstructured data types, following our previous examples, we refer to database items as documents.

\subsection{LLM Prompting Paradigms}

There are three major prompting paradigms for ranking documents using LLMs: pointwise, pairwise, and listwise~\cite{alro}. 

\paragraph{Pointwise prompting} \fn{In this case, an LLM is prompted to assign a ranking score to each document.} 
While LLMs are capable of interpreting user intent for ranking, they are often unreliable at generating zero-short pointwise scores~\cite{setwise}. 
As a result, methods designed for pointwise UDF rankers~\cite{opaque} \fn{require extensive calibration of scores and} are not suitable for semantic top-K queries. 

\paragraph{Pairwise and listwise prompting}
In listwise prompting \fn{and its special case of pairwise prompting for list size of two}, an LLM is instructed to rank input documents based on their relevance to the user's query. 
\fig{fig:framework} shows a prompt template and the corresponding output by the LLM.
The state-of-the-art technique for the semantic top-K operator, LOTUS~\cite{lotus}, resorts to the pairwise prompting method, using general-purpose LLMs. Simply put, LOTUS embeds pairwise ranking prompts in a quicksort algorithm to get a sorted list of documents. This incurs a large number of LLM invocations for large datasets.
On the other hand, listwise prompting is especially promising for its efficiency: it implicitly returns information about a log-linear number of pairwise comparisons. Therefore, the information content per token of listwise ranking is higher than that of pairwise ranking.

\paragraph{Specialized listwise rankers} 
General-purpose LLMs are poor at listwise ranking. 
We show this in our experiments (\sect{sec:expr}). 
They also exhibit a systematic positional bias, where the LLM's output rank is highly correlated with the ordering of documents in the input~\cite{alro}. 
This is because the standard pre- and post-training objectives are not well-suited for ranking tasks.  
A line of work in the IR community fine-tunes small language models as specialized rankers~\cite{rankzephyr,alro,first}, which can rank a small set of documents (e.g., 20) at a time with high accuracy. 
To decrease token usage and improve the efficiency of semantic top-K operators, we adopt these specialized listwise rankers.
Doing so requires adapting existing sorting algorithms to the listwise ranking model.

\paragraph{Challenges of adopting listwise rankers for semantic top-K}
While the aforementioned listwise rankers serve as an effective out-of-the-box ranking kernel for semantic top-K, the standard IR pipelines cannot be naively ported over to the SemOp setting. 
In IR, expensive rankers serve as the final reranker in a model cascade, where its role is to fine-tune the rankings of the most relevant retrieved documents. 
As such, the number of input documents is constrained to a predictably small range. 
In the SemOp setting, however, we are dealing with a much larger and potentially unconstrained number of input documents. 
This necessitates efficient rank aggregation algorithms for combining numerous partial rankings, obtained from listwise rankers, into a total ordering. 
Further, production IR systems have low latency and high throughput targets as well as compute resource constraints. 
In SemOp queries, on the other hand, it is acceptable to spend more time and compute to obtain high-quality results, especially for one-off analytics queries. 
Though latency should be minimized, if possible, without sacrificing accuracy. 
Finally, to deal with the unpredictability of the workload, we need a cost model to tune parameters to minimize latency and to provide a cost estimate for query optimizers. 

\subsection{Rank Aggregation Algorithms}

With an efficient and accurate listwise ranker at our disposal, we turn our attention to designing rank aggregation algorithms. 
These algorithms use the listwise ranker as an oracle and combine the partial ordering of documents, obtained from the ranker's outputs over many invocations, into a top-K solution. 
Table~\ref{tab:algorithms-comparison} presents the characteristics of the studied aggregation algorithms in our pipeline, including an average-case analysis of the number of LLM calls for each algorithm. 

\paragraph{Listwise tournament top-K}
\tselect uses a tournament sort algorithm to find the top-1 document. Using listwise rankings, it organizes data into a graph of pairwise orderings.
To get the top-K results, at each iteration, it removes the winner from the graph and runs another tournament from the documents that could possibly be the runner-up. 
This process is repeated $K$ times. 
\tselect is deterministic and efficient at minimizing LLM calls for small $K$s. 
Indeed, it is the optimal algorithm for $K=1$. 
On the other hand, it is rather complex and is more serial in nature. 

\paragraph{Listwise multi-pivot quickselect/sort}
\qselect and \qsort are our novel generalizations of the standard quickselect/sort algorithms. 
\qselect finds the set of top-$K$ most relevant documents; \qsort fully sorts the top-$K$.
Beyond using listwise rankers for bucket sorting documents in small groups, we further optimize the latencies by using more than one pivot and carefully tuning the pivot count. 
This is the first time this variant of quickselect/sort algorithms -- listwise and with multiple pivots -- have been formally studied. Both algorithms are Las Vegas, simple, and embarrassingly parallel. 

\paragraph{Listwise tournament filter}
Our \tfilter is a fast and simple way to prune the problem size in a lossy manner, as a preprocessing step. 
It randomly shuffles the corpus into small bins, bucket sorts each bin, and lets a tunable constant number of documents in each bin survive for the next round. 
Since \tfilter is computationally efficient, it is a useful pre-filtering step before using the more expensive algorithms. 
It is particularly effective when the corpus is large, and $K$ is small. 
\tfilter is Monte Carlo and easily parallelizable. 

\subsection{\framework Framework}

We design an end-to-end \framework \fn{operator} that combines listwise rankers and the aggregation algorithms.
The data flow plan of \framework is shown in Figure~\ref{fig:framework}. Each path of this graph from the document corpus to the sorted top-K solution is a potential physical query plan for a semantic top-K operator.
It uses a subset of the four algorithms to meet the target recall while minimizing latency for the given problem parameters of corpus size and $K$, under the assumption that the LLM rankings are correct.
By solving for analytical cost and accuracy models, it calculates the optimal parameters for LLM calls and estimates the expected latency, which can be used by the query optimizer. 

\paragraph{Summary of contributions}
\begin{enumerate}
    \item We optimize the execution of semantic top-K operators by using fine-tuned listwise ranker LLMs. 
    \item We study a suite of rank aggregation algorithms with different trade-offs: a deterministic \tselect algorithm; embarrassingly parallel, Las Vegas \qselect and \qsort algorithms; and a Monte Carlo \tfilter. 
    \item We integrate the algorithms into our \framework framework that minimizes latency under recall constraint. 
    \item We implement our framework as a plug-and-play Python library, built on top of an open-source toolkit designed for reranking documents in IR systems, called RankLLM~\cite{rankllm}. The library will be open-sourced by publication.
\end{enumerate}

\begin{figure*}[thb]
    \begin{subfigure}[b]{0.476\textwidth}
        \includegraphics[width=\textwidth]{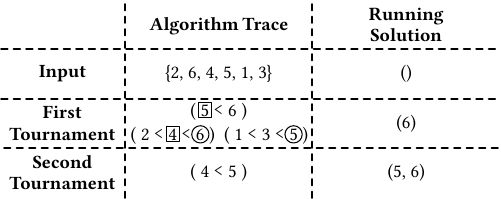}
        \caption{Listwise tournament top-K (\tselect). Circled numbers promote to the next match of the tournament. Squared numbers are candidates for the runner-up tournament.}
        \label{fig:algo:tselect}
    \end{subfigure}
    \hfill
    \begin{subfigure}[b]{0.476\textwidth}
        \includegraphics[width=\textwidth]{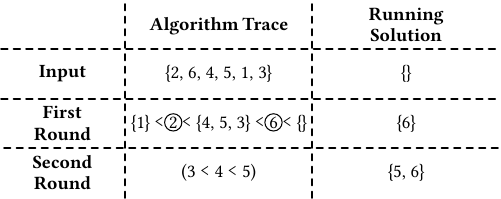}
        \caption{Listwise multi-pivot quickselect (\qselect). Circled numbers are randomly chosen pivots, sorted with respect to each other via a single listwise ranker call.}
        \label{fig:algo:quickselect}
    \end{subfigure}
    
    \begin{subfigure}[b]{0.476\textwidth}
        \includegraphics[width=\textwidth]{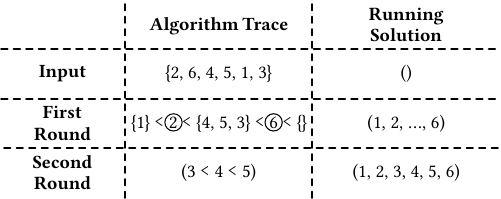}
        \caption{Listwise multi-pivot quicksort (\qsort). Circled numbers are randomly chosen pivots, sorted with respect to each other via a single listwise ranker call.}
        \label{fig:algo:quicksort}
    \end{subfigure}
    \hfill
    \begin{subfigure}[b]{0.476\textwidth}
        \includegraphics[width=\textwidth]{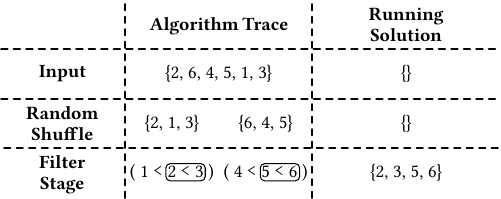}
        \caption{Listwise tournament filter (\tfilter) with $S = 2$. Squared numbers are the survivors of each bucket.\\}
        \label{fig:algo:tfilter}
    \end{subfigure}
    
    \caption{
    A visualization of the four algorithms we study applied to the same toy problem. 
    Across all algorithms, corpus size $N = 6$, list size $L = 3$, and $K = 2$. 
    Numbers in a set $\{a, b\}$ are unsorted whereas numbers in a tuple $(a, b)$ are. The $<$ symbol indicates that an explicit ordering has been established via a listwise ranker. 
    }

    \Description{Algorithm trace and running solution for the four algorithms we study applied to the same toy problem of 2, 6, 4, 5, 1, and 3.}

    \label{fig:algos}
\end{figure*}

\section{Related Work}
\label{sec:prior}

\paragraph{Semantic Top-K} \label{sec:prior:semtopk}
Semantic top-K received scant attention within SemOp research thus far. 
UQE and ThalamusDB support traditional (SQL) top-K, which ranks documents by some numerical attribute~\cite{uqe,thalamusdb}. 
Neither supports semantic top-K. 
Semantic (i.e., vector embedding) similarity search bears resemblance to semantic top-K, but is more rigid and merely captures notions of relevance, whereas semantic top-K queries can capture much more complex priorities using arbitrary natural language. 
LOTUS introduced the first true semantic top-K operator, which we call LOTUS top-K~\cite{lotus}.  
It uses pairwise prompting and quickselect/sort. 

\paragraph{SemBench} \label{sec:prior:sembench}
Recently, \citet{sembench} introduced SemBench, a benchmark for SemOps. 
SemBench includes two ranking queries: Q9 and Q10 on the movies dataset. 
However, both queries -- by default, as shown on their website -- use the LLM to score the movies on a scale, then use traditional methods for sorting. 
In other words, the SemBench queries are of the form
\begin{verbatim}
    SELECT *, "Score the review's sentiment of the 
    movie from 1 to 5." AS score FROM table
    ORDER BY score LIMIT k;
\end{verbatim}
rather than of the form
\begin{verbatim}
    SELECT * FROM table ORDER BY "The review has the most
    positive sentiment." LIMIT k;
\end{verbatim}
which has different interpretations. 
The former can be solved by a semantic mapping operator or a calibrated pointwise prompting method, whereas the latter requires a robust semantic ranker. 
We focus on the latter. 

\paragraph{Other top-K query methods}
There are a number of classical top-K query paradigms that do not map neatly onto semantic top-K. 
Vector similarity search finds the K most similar vectors in a dataset of vectors using an Approximate Nearest Neighbor Search (ANNS) index~\cite{vector_survey}. 
If applicable, vector ANNS search is significantly faster than semantic top-K.
However, its expressivity is limited to notions of similarity in some pre-defined vector space. 
Existing methods for top-K query execution with ML UDF scoring functions assume a precise pointwise oracle~\cite{opaque}. 
As none of these methods map neatly onto semantic top-K, new methods like LOTUS Top-K \cite{lotus} or \framework are required. 

\paragraph{Selection \& top-K sorting algorithms} \label{sec:related:selection}
There are at least three paradigms of selection algorithms that find the $K$\thh largest number in a set of size $N$. 
Since the cost of LLM comparisons is high, we are interested in algorithms with low comparison complexity. 
The quickselect family~\cite{quickselect} uses random pivots to recursively divide the problem into smaller problems that are greater than, or smaller than, the chosen pivot. It achieves the optimal linear time complexity in expectation.
Quicksort is similar to quickselect, but recurses down to both subproblems in each round to fully sort the set in log-linear time.
The combination of the two retrieves and orders the top-K in $O(N + K \log K)$ time.
There is also a family of algorithms based on carefully constructed tournaments~\cite{knuth1998selection}. 
These algorithms typically use a single-elimination tournament bracket to find the winner.
Then, the runner-up must be among the elements directly defeated by the winner, which is decided by another tournament, so forth, for $K$ times. 
The tournament can also be more complex to reduce various overhead.
In general, tournament-based algorithms take $O(N + K \log N)$ time to retrieve and sort the top-K. 
The third family arises from the field of comparison complexity~\cite{dor1999selecting}. 
While they have the lowest amount of comparisons, the overhead is overly complex and expensive, so we do not consider them. 

\paragraph{Multi-pivot quickselect} 
Multi-pivot quickselect algorithms have been studied in the past, but in the context of pairwise comparison oracles. 
\citet{wild2012average} studied Java's dual-pivot quicksort. 
\citet{holmgren2025probabilistic} generalized the model to an arbitrary number of pivots. 
These results do not easily translate to listwise oracles, so we perform independent mathematical analysis in \sect{sec:analysis}.

\paragraph{Listwise Rankers} \label{sec:prior:ranker}
IR and NLP researchers investigated how to make zero-shot listwise ranking with LLMs more accurate and efficient. 
RankZephyr~\cite{rankzephyr}, a fine-tuned version of Zephyr~\cite{tunstall2023zephyr} -- which is itself instruction-tuned from Llama -- distills the ranking capabilities of GPT-4 into a small language model. 
It employs techniques such as input data shuffling, variable window sizes (2-20), and hard negative mining~\cite{rankzephyr}. 
It is the de facto open-source listwise ranking model, implemented in the RankLLM library~\cite{rankllm}. 
We use RankZephyr in our experiments. 
There are also more recent and actively researched improvements to listwise rankers (see \sect{sec:discussion}). 
\framework is agnostic to the choice of listwise ranker; \framework can easily adopt improved listwise rankers. 
\section{Method} \label{sec:method}

We now describe our rank aggregation algorithms and the \framework framework. 
We describe our algorithms at a high level (\sect{sec:method:algo}), then discuss practical optimizations (\sect{sec:method:optim}), and the holistic framework (\sect{sec:method:framework}). 
Refer to \fig{fig:algos} for concrete demonstrations on a toy problem and \tbl{tab:general-vars} for mathematical notations. 

\begin{table}
    \centering
    \begin{tabularx}{\columnwidth}{|c|X|}
    \hline
    Variable & Definition \\
    \hline
    $\mathcal{D}$ & Document corpus\\
    $N$ & Corpus size\\
    $K$ & Number of documents to return\\
    $L$ & List size; number of documents in each LLM call\\
    $P$ & Pivot count (\qselect, \qsort)\\
    $P'$ & Physical pivot count; \# pivots in each LLM call\\
    $S$ & Number of survivors per bin (\tfilter) \\
    \hline
    \end{tabularx}
    \caption{Variables used in algorithms.}
    \label{tab:general-vars}
\end{table}

%\begin{table}[h!]
%    \centering
%    \begin{tabular}{cc}
%        $p$ & Number of pivots\\
%        $x$ & Maximum number of pivots placed into each LLM call\\
%        $w$ & Maximum number of documents placed into a single LLM call\\
%        $k$ & Number of documents to return\\
%        $M_i$ & Maybe set of iteration i\\
%        $D$ & Definitely set\\
%    \end{tabular}
%    \label{tab:variables}
%\end{table}

\subsection{Algorithm Descriptions} \label{sec:method:algo} 

Our semantic top-K operator starts with an optional lossy filtering algorithm (Algorithm~\ref{alg:tfilter}) that is a simple method for reducing the problem size.
%It returns a subset of documents from the corpus that likely contains the majority of the ground truth top-K.
This is followed by one of two plans: 1) a deterministic tournament-based method (\alg{sec:method:algo:tselect}), or 2) a generalization of quickselect and quicksort with listwise comparisons and multiple pivots (\alg{sec:method:algo:qselect}, \alg{sec:method:algo:qsort}).

All of our algorithms use listwise prompting methods, as shown in \fig{fig:framework}.
Each prompt includes a numbered list of documents and a user query.
%(For example, ``'')
The ranker then returns a sorted list of document IDs in descending order of relevance to the user query.

\subsubsection{Tournament top-K} \label{sec:method:algo:tselect}

\begin{algorithm}
    \caption{Listwise Tournament Top-K (\tselect)}
    \begin{algorithmic}[1]
        \Ensure Ordered Top-K solution
        \State $\text{graph} \gets (\mathcal{D}, \emptyset)$ \Comment{Graph of partial comparisons}
        \State $\text{solution} \gets []$
        \State Shuffle $\mathcal{D}$ into bins of size $L$ each
        \For{$k = 1$ to $K$}
            \State $\text{winner} \gets$ \textsc{RunTournament}($\text{graph}$)
            \State Append $\text{winner}$ to solution
            \State Remove $\text{winner}$ and all associated edges from $\text{graph}$
        \EndFor
        \State \Return $S$
    \end{algorithmic}
    \label{alg:tselect}
\end{algorithm}

\begin{algorithm}
    \caption{\textsc{RunTournament}}
    \begin{algorithmic}[1]
        \Require Graph $G = (V, E)$ of partial comparisons, bin size $L$
        \Ensure Winner element from $V$
        \State $\text{candidates} \gets$ \{$v \in V : $ no outgoing edges in $E$\}
        \While{$|\text{candidates}| > 1$}
            \State Shuffle $\text{candidates}$ into bins of size $L$ each
            \State $\text{round\_winners} \gets []$
            \For{each bin}
                \State Run listwise sort on bin
                \State Update $G$ with new comparison edges from sort
                \State Append top element from sorted $B$ to round\_winners
            \EndFor
            \State $\text{candidates} \gets \text{round\_winners}$
        \EndWhile
        \State \Return Head of $G$ \Comment{$G$ is now a tree}
    \end{algorithmic}
    \label{alg:runtournament}
\end{algorithm}

\begin{algorithm}
    \caption{Listwise Multi-Pivot Quickselect (\qselect)}
    \begin{algorithmic}[1]
        \Require Logical pivot count $P$, physical pivot count $P' \leq P < L$
        \Ensure Unordered top-$K$ solution set
        \State $\text{solution} \gets \emptyset$
        \State $\text{candidates} \gets \mathcal{D}$
        
        \While{$|\text{solution}| < K$}
            \State pivots $\gets$ random sample of $P$ documents from $\mathcal{D}$ %Randomly sample $P$ pivots from candidates
            \State pivots $\gets$ run listwise LLM sort on pivots %Run listwise oracle sort on pivots
            \State $\text{buckets} \gets$ array of $P + 1$ empty lists
    
            \State buckets $\gets$ bucket sort candidates into buckets via listwise sort, where each sort contains $P'$ pivots and $L - P'$ non-pivots
            
            \State $j \gets$ smallest $j$ such that $\sum_{i=0}^{j} |\text{buckets}[i]| \geq K - |\text{solution}|$
            \State Append elements from buckets[$0\ldots j-1$] to solution
            \State $\text{candidates} \gets \text{buckets}[j]$
            \State $K \gets K - |\text{solution}|$
        \EndWhile
        
        \State \Return solution
    \end{algorithmic}
    \label{alg:qselect}
\end{algorithm}
\begin{algorithm}
    \caption{Listwise Multi-Pivot Quicksort (\qsort)}
    \begin{algorithmic}[1]
        \Require Logical pivot count $P$, physical pivot count $P' \leq P < L$
        \Ensure Fully sorted list
        \State $\text{result} \gets []$
        \State $\text{candidates} \gets \mathcal{D}$
        
        \If{$|\text{candidates}| \leq L$}
            \State Run listwise %oracle 
            sort on candidates
            \State \Return sorted candidates
        \EndIf
        
        %\State Randomly sample $P$ pivots from candidates
        %\State Run listwise oracle sort on pivots
        %\State $\text{buckets} \gets$ array of $P + 1$ empty lists
        %\State Bucket sort candidates into buckets via listwise sort, where each sort contains $P'$ pivots and $L - P'$ non-pivots
        \State pivots $\gets$ random sample of $P$ documents from $\mathcal{D}$ %Randomly sample $P$ pivots from candidates
        \State pivots $\gets$ run listwise LLM sort on pivots %Run listwise oracle sort on pivots
        \State $\text{buckets} \gets$ array of $P + 1$ empty lists
        \State buckets $\gets$ bucket sort candidates into buckets via listwise sort, where each sort contains $P'$ pivots and $L - P'$ non-pivots
        \For{$i = 0$ to $P$}
            \State $\text{sorted\_bucket} \gets$ \qsort(buckets[$i$], $P$, $P'$)
            \State Append sorted\_bucket to result
            \If{$i < P$}
                \State Append sorted\_pivots[$i$] to result 
            \EndIf
        \EndFor
        
        \State \Return result
    \end{algorithmic}
    \label{alg:qsort}
\end{algorithm}
\begin{algorithm}
    \caption{Listwise Tournament Filter (\tfilter)}
    \begin{algorithmic}[1]
        \Require Survivor count $S$
        \Ensure Filtered subset of documents
        \State $\text{solution} \gets \emptyset$
        \State Shuffle $\mathcal{D}$ into bins of size $L$ each
        \For{each bin}
            \State Run listwise sort on bin
            \State Add the top $S$ documents in bin to solution
        \EndFor
        
        \State \Return solution
    \end{algorithmic}
    \label{alg:tfilter}
\end{algorithm}

Algorithm~\ref{alg:tselect} is a deterministic, tournament-based method for finding the ordered top-K solution. 
It maintains a graph of partial comparisons, where the vertices are documents and a directed edge from $u$ to $v$ indicates that document $v$ is ranked higher than document $u$. 
\tselect repeatedly runs a subroutine \textsc{RunTournament} that finds candidates for the next winner.
The tournaments are similar to that of a standard tournament-based algorithm (\sect{sec:related:selection}), except that each ``match'' in the single-elimination bracket has $L$ (list size) players and one winner, as opposed to just two players. 

Winners advance up the graph until the overall winner is found at the root. 
The winner is added to the solution and is removed from the graph. 
\tselect then runs another tournament bracket to find the runner-up. 
The candidates for this tournament are documents with no outgoing edges, i.e., documents defeated by the first winner. 
This process is repeated $K$ times.

\fig{fig:algo:tselect} shows the trace of the algorithm on an input with six documents and two winners. 
The first tournament has three matches, with up to three players ranked in each match, and the sole winner of 6. 
The second tournament has two candidates that are fully sorted in a single match, producing the runner-up of 5.

\subsubsection{Listwise multi-pivot quickselect} \label{sec:method:algo:qselect}

Algorithm~\ref{alg:qselect} is a generalization of quickselect with listwise sorting and multiple pivots. 
In each round, Algorithm~\ref{alg:qselect} randomly samples $P$ documents as pivots. 
We constrain $P$ to be less than $L$, so the $P$ pivots can be sorted in a single LLM call. 
We then bucket sort the remaining non-pivot documents into $P + 1$ buckets. 
The elements of each bucket are partially ordered with respect to the $P$ pivots.
Of the $P + 1$ buckets, one is guaranteed to contain the $K$\thh best document; it becomes the next candidate set. 
Suppose we have $B_1, \ldots, B_{P+1}$ buckets. Documents in $B_i$ are smaller than the $i$\thh smallest pivot. 
\qselect selects the next candidate set $B_j$ with smallest $j$ such that $K \leq \sum_{i=1}^j |B_i|$.
Documents in buckets $B_1, \ldots, B_{j-1}$ are guaranteed to be contained in the top-K and are added to the running solution. 
Next, \qselect recurses down to the documents in $B_j$, with $K$ updated accordingly. 
This process continues until we have exactly $K$ documents in the solution set.

See \fig{fig:algo:quickselect} for a fully worked out toy instance of \qselect. 
In the first round, two random pivots $\{2, 6\}$ are selected, and the remaining candidates $\{1, 4, 5, 3\}$ are sorted into three buckets using four listwise ranker calls, each call containing two pivots and one non-pivot. 
This guarantees the top bucket (which is empty, containing no documents greater than the pivot 6) and the pivot 6 to be contained in the top-2 solution. 
The second round sorts $\{4, 5, 3\}$ by using a single listwise ranker call, as now the problem size has become equal or less than $L$. 
This guarantees 5 to be contained in the top-2 solution, and the algorithm terminates. 
We may, in general, include fewer than $P$ pivots in each LLM call to minimize the number of LLM calls. 
We denote that number as $P'$ and call it the physical pivot count. 
If $P' < P$, then for each group of $L - P'$ non-pivot candidates, we issue multiple LLM calls with disjoint groups of pivots to fully bucket sort them. 

\subsubsection{Listwise multi-pivot quicksort} \label{sec:method:algo:qsort}

Algorithm~\ref{alg:qsort} is a generalization of quicksort analogous to \qselect. 
The only difference is that it recurses down to each of the $P + 1$ subproblems to return a fully sorted list of documents. 
The algorithm trace in \fig{fig:algo:quicksort} is similar to \fig{fig:algo:quickselect}, except that the solution is an ordered tuple rather than an unordered set.

\subsubsection{Listwise tournament filtering} \label{sec:method:algo:tfilter}

Algorithm~\ref{alg:tfilter} is a simple method for reducing the problem size, returning a subset that likely contains the majority of the ground truth top-K solution. 
\tfilter shuffles the corpus into bins of size $L$ each, listwise sorts each bin with a single LLM call, and adds the top $S$ documents of each bin to the solution. 
With random shuffling and properly tuned $S$, the solution set is highly likely to contain the majority of the ground truth top-K. 
\tfilter can be nested several times in succession, hence, the name \emph{tournament}. 
However, in practice, we only perform a single round of filtering. 
\fig{fig:algo:tfilter} shows a failure case for this algorithm: it has a 100\% recall for the top-1 and top-2, but has a lower recall for top-3 as its solution does not contain the element 4.

\subsection{Practical Optimizations} \label{sec:method:optim}

\paragraph{Data parallelism}
All algorithms discussed in \sect{sec:method:algo} allow for extreme levels of data parallelism: the initial tournament in \tselect, bucket sort in \qselect and \qsort, and the entirety of \tfilter. 
Our implementation detects the number of available GPUs allotted to our system and exploits data parallelism by issuing parallel calls to multiple instances of listwise rankers. 

\paragraph{Embedding-based pivot selection}
LOTUS \cite{lotus} introduced a heuristic for choosing the pivots based on the embedding distance between the documents and the query~\cite{lotus}. 
The key idea is to use the distance as a proxy metric for the correct rank. 
If the LLM's ranking of the documents is similar to that given by embedding distances to the user query, we can choose nearly optimal pivots. 
If not, the overhead of producing embeddings is negligible, and the pivot selection becomes effectively random. 
A purely adversarial relationship between the LLM's rankings and the embedding-based proxy rankings is unlikely. 
We adopt and generalize this heuristic to multiple pivots. 
For \qselect, we first sort the candidate documents based on their embeddings' similarity to the query. 
Then, we choose the $K$\thh, $K+1$\thh, \ldots, $K+P-1$\thh documents as the pivot.
This makes it likely that we choose pivots whose correct rank is close to the threshold that divides documents into top-K and not top-K, thereby further decreasing the problem size in each round. 
For \qsort, we choose the pivots such that they evenly divide the corpus with respect to embedding distance to the query. 
Evenly dividing the subproblems minimizes latency by reducing the influence of the superlinear time complexity of quicksort.

\paragraph{Ordering LLM calls and early stopping} 
The number of pivots included in each LLM call (the physical pivot count $P'$) may be smaller than the number of pivots as specified in the algorithm (the logical pivot count $P$). 
This sometimes allows us to prune some number of LLM calls.
We first compare candidate documents against the largest $P'$ pivots. 
If this gives us $K$ or more documents that are ranked higher than the $P'$ pivots, we know we have found a superset of the top-K solution. 
Hence, we may immediately proceed to the next round.  
Only if we have not found at least $K$ documents that are ranked higher than the $P'$ pivots, then do we proceed to compare the candidates against the next largest group of $P'$ pivots.
Consider the toy example in \fig{fig:algo:quickselect}. 
In the first round, we have two pivots $\{2, 6\}$. 
We can sort the candidates $\{1, 4, 5, 3\}$ with respect to the larger pivot 6 first, obtaining the partial ordering $\{1, 4, 5, 3\} < 6 < \emptyset$. 
If $K = 1$, then we know that $\{6\} \cup \emptyset$ is guaranteed to contain the top-1 solution. 
Hence, early stopping reduces the total number of LLM calls. 
If $K > 1$, then $\{6\} \cup \emptyset$ does not contain the top-K, so we need to perform additional bucket sort. 
We empirically study when early stopping is effective in \sect{sec:expr:result:optim}.

\paragraph{Resolving and mitigating self-contradictions} 
In our theoretical analysis (\sect{sec:analysis}), we assume that the listwise rankers are perfect. 
However, real-life listwise rankers are not. 
Suppose we have two pivots $P_1, P_2$.
The ranker might first sort the pivots as $P_1 < P_2$. 
Then, when asking the ranker to sort the pivots and a non-pivot candidate $P_1, P_2, C_1$, it may return a ranking like $P_2 < C_1 < P_1$, which contradicts itself. 
If so, we enforce self-consistency by interpreting the ranking as $P_1 < C_1 < P_2$. 
This is a blunt policy that we adopt for latency reduction. 
To reduce the probability of self-contradictions, we include the ``correct'' ordering of the pivots in the prompt. 

\paragraph{Fine-tuning the top-few rankings} \label{sec:method:optim:finetune}

Our \qsort algorithm has been optimized to minimize latency for sorting a large corpus.
If $K$ is very small, or for the top-few most relevant documents, it makes sense to sort them using a more expensive algorithm.  
Therefore, we use \qsort for the top-$K$, then use pairwise quicksort to rerank the top-$10$. 
This improves the NDCG of our solution with a small constant-bounded cost to latency (\sect{sec:expr:result:optim}). 

\subsection{\framework Framework} \label{sec:method:framework}

Our full framework, as shown in \fig{fig:framework}, is a physical operator for semantic top-K with four plans. 
The plans may or may not use \tfilter. 
It then uses one of two aggregation algorithms: \tselect, or \qselect with \qsort. 
We test the four plans as independent algorithms in our experiments (\sect{sec:expr}). 
In our \framework library, we abstract away the plans with an API for plan optimization and cost estimation. 
This is to make it more useful as a database component. 
We are developing this feature to be available when our library is open-sourced. 

The plan optimizer is tasked with choosing a physical plan and tunable parameters ($P_{select}$, $P_{sort}$, $S$) where $P_{select}$ denotes the logical pivot count of \qselect and $P_{sort}$ the logical pivot count of \qsort. 
It is also given query-specific parameters $N$ and $K$; a fixed $L$ for the listwise ranker; and a minimum recall target $R_{\text{target}}$.
This forms an optimization problem. 

\begin{mini}
    {\substack{S, P_{\text{select}}, P_{\text{sort}},\\ \text{use\_filter}, \text{ agg}}}
    {\mathrm{Cost}_{\text{total}}(N, K, L, S, P_{\text{select}}, P_{\text{sort}}, \text{use\_filter}, \text{agg}) \label{eq:optim}}
    {}
    {}
    \addConstraint{S, P_{\text{select}}, P_{\text{sort}}}{\in \mathbb{N}}
    \addConstraint{1 \le S, P_{\text{select}}, P_{\text{sort}}}{< L}
    \addConstraint{\text{use\_filter}}{\in \{\text{True}, \text{False}\}}
    \addConstraint{\text{agg}}{\in \{\text{LTTopK}, \text{LMPQ}\}}
    \addConstraint{R_{\text{target}}}{\ge \mathrm{Recall}_{\text{LTFilter}}(N, K, L, S)}
\end{mini}
Note that $P'$ is not in the optimization problem as we constrain $P' = P$ when auto-tuning, which is optimal for a wide range of $P$.
We first define $N_{\text{filtered}}$, the number of documents remaining after the optional filter. 
\begin{equation}
    N_{\text{filtered}} = \begin{cases}
        S \cdot\left\lceil \frac{N}{L} \right\rceil & \text{if use\_filter}\\
        N &\text{otherwise}
    \end{cases}
    \label{eq:n-filtered}
\end{equation}
Now, the objective function $\mathrm{Cost}_{\text{total}}$ is as follows. 
\begin{align}
    &\mathrm{Cost}_{\text{total}}(\cdot) = \mathrm{Cost}_{\text{LTFilter}}(N, L) \times \mathbb{I}[\text{use\_filter}] \nonumber\\
    &\quad+\mathrm{Cost}_{\text{LTTopK}}(N_{\text{filtered}}, K, L) \times \mathbb{I}[\text{agg} = \text{LTTopK}] \nonumber\\
    &\quad+\mathrm{Cost}_{\text{LMPQSelect}}(N_{\text{filtered}}, K, L, P_{\text{select}}) \times \mathbb{I}[\text{agg}=\text{LMPQ}] \nonumber\\
    &\quad+ \mathrm{Cost}_{\text{LMPQSort}} (K, L, P_{\text{sort}}) \times \mathbb{I}[\text{agg}=\text{LMPQ}]
    \label{eq:cost-total}
\end{align}
The subroutines here are defined and proven in \sect{sec:analysis}. 
Note that some cost models include big-O terms; we use coefficients fit from empirical data, as described in \sect{sec:theory:mc}. 
In the special case of full sort, we always omit \tfilter and also omit \qselect from the \lmpq plan.  

We solve the optimization problem as follows. 
We first solve for the minimum $S$ that would meet the target recall by binary search. 
Now the optimizer is constrained to using the filter with that minimum $S$, or skipping the filter. 
Next, we find the optimal $P_{\text{select}}, P_{\text{sort}}$ given $L$ by using \cor{cor:optim-pivot-sort} and \eq{eq:optim-pivot-select}.  
This fixes all parameters of $\text{Cost}_{\text{total}}$ except the physical plan. 
We then enumerate the cost of all four plans and choose the one with the minimal cost. 

\section{Analysis} \label{sec:analysis}

\begin{figure*}[thb]
    \centering
        \begin{subfigure}[b]{0.49\textwidth}
        \centering
        \includegraphics[width=\textwidth]{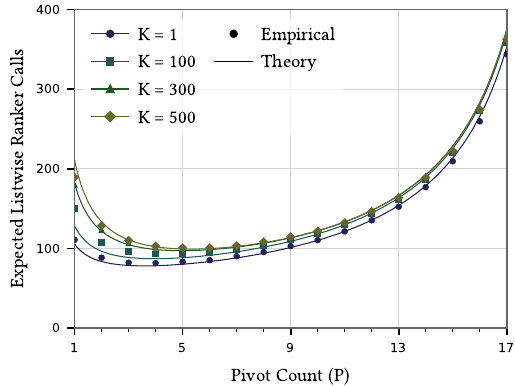}
        \caption{\qselect}
        \Description{The relationship between the expected listwise ranker calls and pivot count for LMPQSelect is a J-shaped curve with slight difference based on K.}
        \label{fig:simu:qselect}
    \end{subfigure}
    \hfill
    \begin{subfigure}[b]{0.49\textwidth}
        \centering
        \includegraphics[width=\textwidth]{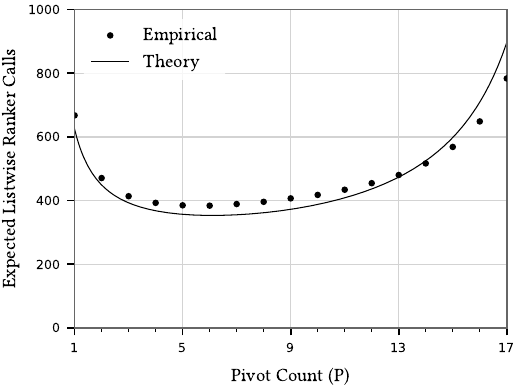}
        \caption{\qsort}
        \Description{The relationship between the expected listwise ranker calls and pivot count for LMPQSort is a U-shaped curve.}
        \label{fig:simu:qsort}
    \end{subfigure}
    \caption{
    Ideal Monte-Carlo simulation of \qselect and \qsort compared to the asymptotic theoretical estimate. 
    Cost is measured in terms of the number of listwise ranker calls; rankers are assumed to be perfect. 
    5000 trials for each parameter configuration. 
    For quickselect, we test varying values of $K$, where $\psi = K/N$ ratio ranges from 0.001 to 0.5.
    Note that y axis is truncated at very large values.}
    \label{fig:simu}
\end{figure*}

In this section, we concretize the cost and recall functions used without definition in \eq{eq:optim} and \eq{eq:cost-total}. 

\paragraph{Simplifying assumptions}
We assume that the listwise ranker is perfect, in that it always returns the ground truth ranking of up to $L$ documents. 
While real-life listwise rankers are noisy, modeling the noise requires making too strong of assumptions. 
We also assume that shuffling and pivot selection are done uniformly randomly. 
In practice, we use embedding-based pivot selection (\sect{sec:method:optim}), which improves average-case latency. 

\paragraph{Cost model}
We use the number of LLM calls as the cost measure. 
Due to the expensive nature of LLM calls, the non-LLM overhead is negligible. 
This contrasts with conventional analysis of selection and sorting algorithms~\cite{wild2012average,quicksort}. 

\subsection{Analysis Results} \label{sec:theory:result}

We now present the precise cost and recall models. 
For brevity, long proofs are in the appendix (\sect{sec:appendix:proof1}, \sect{sec:appendix:proof2}). 

\subsubsection{Tournament Top-K} \label{sec:theory:result:tourney}

\begin{remark}[Cost of listwise tournament top-K]
    For large $N$: 
    \begin{equation}
        \mathrm{Cost}_{\text{LTTopK}}(N, K, L) \approx \frac{N + (K-1) \log_L N}{L-1}
    \end{equation}
    \label{rem:tselect-complexity}
\end{remark}

\begin{proof}
    $\frac{N}{L-1}$ cost is incurred in the initial tournament: it eliminates $N - 1 \approx N$ candidates in total, and each LLM call eliminates $L - 1$. 
    Each subsequent tournament uses at most $\log_L N$ candidates --- those defeated by the previous winner --- and thus costs $\frac{\log_L N}{L-1}$ LLM calls each. \fixed{There are $K-1$ such tournaments.}
\end{proof}

\subsubsection{Listwise Multi-Pivot Quicksort} \label{sec:theory:result:qsort}

Like other quicksort methods, \qsort has a log-linear comparison complexity. 
We compute the dominating constant factor while simplifying smaller terms. 

\begin{theorem}[Expected cost of listwise quicksort]
    Constrain $P = P'$. 
    Then for sufficiently large $N$:
    \begin{equation}
        \mathrm{Cost}_{\text{LMPQSort}}(N, L, P) = \frac{1}{(L-P) \log(P+1)} \cdot N \log N + O(N).
        \label{eq:quicksort-expectation}
    \end{equation}
    \label{thm:quicksort-expectation}
\end{theorem}

The $1/(L-P)$ term arises from the fact that $L-P$ non-pivot elements are included in each LLM call. 
It incentivizes $P$ to be small. 
The $1/\log(P+1)$ term captures that a larger pivot count reduces the number of expected rounds, as the subproblem gets divided in a more fine-grained manner. 
It pushes $P$ to be large. 

\paragraph{Solving for the optimal pivot count}
By setting the derivative of the leading coefficient to zero, we obtain the optimal $P$ w.r.t. $L$. 
\begin{corollary}
    \eq{eq:quicksort-expectation} is minimized when
    $$P^* = \frac{\log(P+1) + P - L}{[(L-P) \cdot \log(P+1)]^2}.$$
    \label{cor:optim-pivot-sort}
\end{corollary}
For $L = 20$, $P^* \approx 6.1$ which cleanly rounds down to $P = 6$. 

\subsubsection{Listwise Multi-Pivot Quickselect} \label{sec:theory:result:qselect}

Like other quickselect variants, \qselect is a linear-time algorithm. 
Computing the dominating coefficient, however, is more challenging. 
This is because the cost model is not only dependent on $N$, but also on $K$. 
Consider the ratio of $K$ to $N$, denoted $\psi_i = K_i/N_i$ for the $i$\thh round of quickselect. 
While $\psi_1$ is known a priori, $\psi_2$ and henceforth are unknown random variables. 
\fixed{This breaks self-similarity, which is used in \thm{thm:quicksort-expectation} to reduce the problem to a fixed recurrence relation.}
We assume self-similarity for simplicity; it lines up well with simulation data. 

\begin{theorem}[Expected complexity of listwise quickselect]
    Constrain $P = P'$. 
    Let $\psi = K/N$. 
    Assume self-similarity. 
    Then, for sufficiently large $N$:
    \begin{align}
        &\mathrm{Cost}_{\text{LMPQSelect}}(N, K, L, P) \nonumber\\
        &\quad\quad = \frac{N(P+1)}{(L-P)\left(P - 1 + \psi^{P+1} + (1-\psi)^{P+1}\right)} + O(1).
        \label{eq:quickselect-expectation}
    \end{align}
    \label{thm:quickselect-expectation}
\end{theorem}

The $1/(L-P)$ term is analogous to that in \thm{thm:quicksort-expectation} and incentivizes fewer pivots, whereas the other multiplicand incentivizes more pivots. 

\paragraph{Solving for the optimal pivot count.}
The derivative of the leading coefficient for \eq{eq:quickselect-expectation} is quite complex, so solving for $P^*$ via brute force search is easier. 
The exception is the common case when $\psi \approx 0$, then:  
\begin{equation}
    P^*_{\text{LMPQSelect}} = -1 + \sqrt{1 + L}.
    \label{eq:optim-pivot-select}
\end{equation}
For $L = 20$, $P^* \approx 3.58$ which rounds up to $P = 4$. 

\subsubsection{Tournament Filter} \label{sec:theory:result:filter}

Trivially, 
\begin{equation}
    \mathrm{Cost}_{\text{LTFilter}} = \left\lceil \frac{N}{L} \right\rceil. 
\end{equation}

\begin{remark}[Expected recall of listwise tournament filter]
    Let $\lambda = KL/N$ denote the expected number of top-$K$ elements per bin. Under a Poisson approximation for the bin occupancies:
    \begin{equation}
        \mathrm{Recall}_{\text{LTFilter}} = P(M \leq S-1) + \frac{S}{\lambda}P(M \geq S)
    \end{equation}
    where $M \sim \mathrm{Poisson}(\lambda)$.
    \label{rem:tfilter-recall}
\end{remark}

\subsection{\qselect/\qsort Simulations} \label{sec:theory:mc}

\thm{thm:quicksort-expectation} and \thm{thm:quickselect-expectation} can be challenging to parse. 
The theorems also don't guarantee that the asymptotic analysis matches reality. 
To help understand the key results, we run Monte Carlo simulations of \qsort and \qselect under the ideal assumptions made above, for reasonably small $N$ of 1,000. 
For quicksort, we fit a coefficient of 0.1 for the $O(N)$ term from empirical data with a variety of $N$ values ranging from $N = 100$ to $N = 100,000$; for quickselect, we find that omitting the $O(1)$ term fits the data well already. 
The results are shown in \fig{fig:simu}. 

\paragraph{Pivot count tuning}
The theoretical predictions for the optimal pivot count ($P^*_{\text{LMPQSelect}} = 4, P^*_{\text{LMPQSort}} = 6$) match the simulation data. 
Therefore, our modeling is useful for pivot count tuning with realistic $N$. 
Tuning pivot count gives us a latency reduction between 10\% and 40\%. 

\paragraph{Quickselect}
\fig{fig:simu:qselect} shows the J-shape trade-off curve between small and large $P$ as predicted by \thm{thm:quickselect-expectation} . 
Furthermore, the model predicts the empirical data well with $O(1) = 0$. 
There is symmetry around $\psi = 1/2$, and when $\psi$ is extreme, the cost is cheaper. 

\paragraph{Quicksort}
\fig{fig:simu:qsort} shows an analogous U-shape tradeoff curve for \qsort. 
A simple model of $O(N) = 0.1N$ reliably predicts the empirical cost. 
\section{Experiments} \label{sec:expr}

%\subsection{Preliminaries} \label{sec:expr:prelim}
\subsection{Setup}
\subsubsection{Datasets \& Tasks}
We test our algorithms and baselines on two datasets~\cite{lotus,sembench}. 
We denote the number of relevant documents per query as $R$. 

\paragraph{SciFact}
SciFact~\cite{wadden2020fact} is a standard IR benchmark dataset containing 5,183 PubMed abstracts and 1,409 queries -- scientific claims that can be verified with one or more documents. 
On average, $R = 1.1$, and the majority of queries have just one relevant document. 
We evaluate our algorithms on the first 25 test queries. 
\lotustopk reports results on only a subset of 100 most relevant documents. To show scalability, we use all 5,183 documents in our experiments. 

\paragraph{LLM-as-judge augmentation}
SciFact is a sparse dataset with $R = 1.1$. 
To test whether the observed trends replicate in a denser dataset with $R \ge K$, we establish a gold-standard ranking using an expensive model cascade. 
It uses a diverse ensemble of retrievers and rerankers followed by a re-ranking of the top-50 using an oracle-class LLM (GPT-4o), pairwise prompting, and Elo ratings~\cite{elo}. (See \sect{sec:appendix:augmentation} for details.)

\paragraph{Movies}
Movies is one of the datasets in SemBench. (see \sect{sec:prior:sembench} for more on SemBench.) 
It includes 1M test reviews of movies. 
We evaluate our methods on query 9 of SemBench, which trims the dataset size to 256 reviews of the same movie and runs a full sort, i.e. $K = N = 256$.

\subsubsection{Algorithms}

\paragraph{Our algorithms}
We test all four plans -- with or without \tfilter, and \tselect or \lmpq as the rank aggregator -- in \framework with all empirical tricks (\sect{sec:method:optim}) and tuned parameters. 
The survivor count of \tfilter varies among $S = \{1, 5, 10, 15\}$. 
By default, we use 16 pivots with $P' = 2$ and with early stopping enabled, which empirically slightly edges out the theoretically predicted pivot counts of $P^*_{\text{sort}} = 6$ and $P^*_{\text{select}} = 4$ (\fig{fig:expr:heur:early}). 
As the difference is minuscule, the performance of empirically tuned parameters is likely a minor optimization compared to theoretically auto-tuned parameters, if \framework is to be applied to novel settings. 

\paragraph{Pairwise baseline}
\lotustopk is the prior art implemented in LOTUS. 
It combines pairwise prompting quickselect with quicksort on general-purpose LLMs, and uses the embedding distance based pivot selection heuristic. 
We run \lotustopk using the official LOTUS code base. 
We also test \pairwise, which is \qselect and \qsort with one pivot and list size of two. 
It is the same algorithm as \lotustopk that uses listwise rankers instead of general-purpose LLMs. 

\paragraph{Pointwise baseline}
We implement a \pointwise baseline using LOTUS' semantic mapping operator. 
We ask a general-purpose LLM to give a numeric score between 1 and some maximum score $M$. ($M = 5,000$ for SciFact; $M = 5$ for Movies.) 
We then parse the score as a float, where invalid values are imputed as $M / 2$. 
We then sort the corpus by the scores and return the sorted top-K.

\subsubsection{Models}

We use RankZephyr~\cite{rankzephyr} as the default listwise ranker (see \sect{sec:prior:ranker} for more on RankZephyr) as well as Zephyr-7B~\cite{tunstall2023zephyr} and Qwen3-8B~\cite{qwen3} as general-purpose models of comparable parameter class. 
Vector embeddings for the pivot selection heuristic use all-MiniLM-L6-v2.

\subsubsection{Objectives}
We evaluate selection and sorting algorithms with $K < N$ using Recall@K and NDCG@K, respectively. 
For the special case of full sort, we use Spearman's correlation.

\subsubsection{Implementation}
We implement our algorithms as a standalone Python library. 
It leverages the RankLLM library~\cite{rankllm} that encapsulates the listwise ranker LLM calls. Our listwise ranker of choice is RankZephyr~\cite{rankzephyr}.
It supports our algorithms (\tselect, \lmpq, \tfilter) as well as faithful re-implementations of pairwise quickselect/sort~\cite{lotus} ranking baseline.
We have rankLLM utilize its vLLM~\cite{vllm} implementation to serve RankZephyr, Zephyr-7B, and Qwen3-8B for our algorithms.

\subsubsection{System Details}
We run all experiments on a server with an Intel(R) Xeon(R) Silver 4410Y CPU and four NVIDIA RTX A5000 GPUs that each host an instance of RankLLM.

\subsection{Results}

\subsubsection{End-to-End Comparison}
\begin{figure*}[thb]
    \centering
    \includegraphics[width=\textwidth]{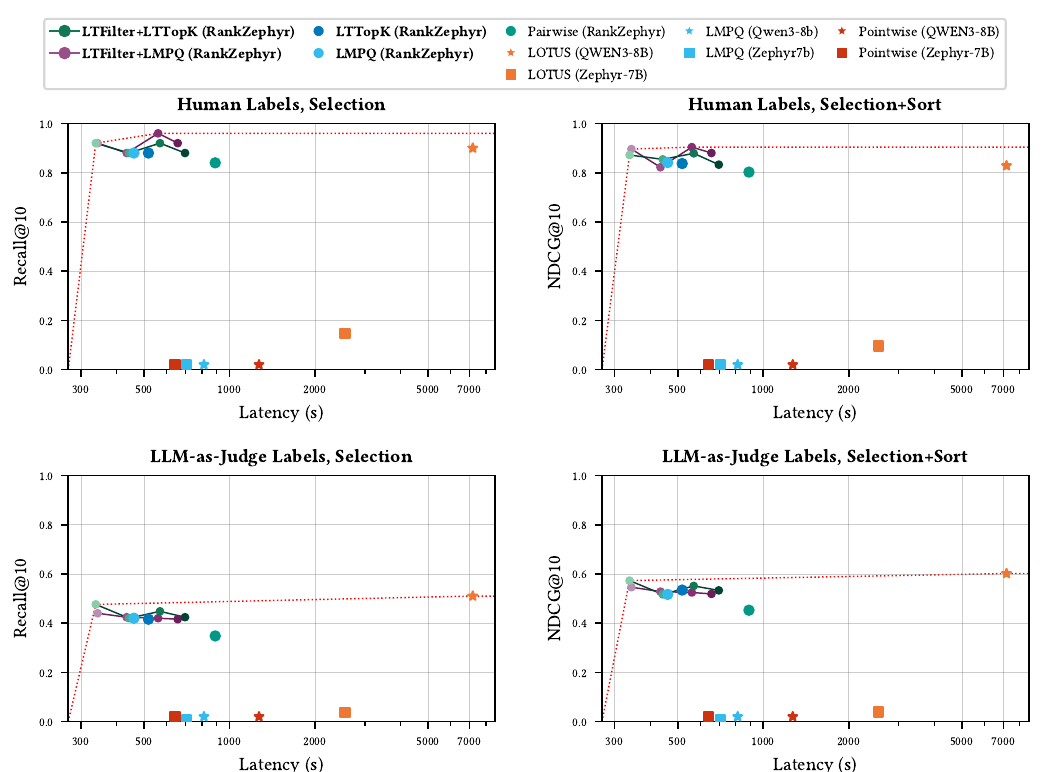}
    \caption{
    End-to-end comparison of our methods (LTFilter+LMPQ, Tournament Top-K, LMPQ) with baseline methods on SciFact. 
    All methods were tested on 25 test queries of the SciFact dataset with $N=5,183$ and $K=10$ for the selection and selection+sorting problems, w.r.t. mean accuracy (recall, NDCG) and latency. 
    Our methods dominate the Pareto frontier with the exception of \lotustopk with Qwen3-8B on the extremely expensive end of the spectrum. 
    Note the log scale on the $x$-axis. 
    For \tfilter, $S$ = 1, 5, 10, and 15 were tested with a darker shade indicating a larger $S$ value.
    }
    \Description{The plot is subdivided into four subplots for human vs llm-as-judge-labels and selection vs selection and sorting tasks.}
    \label{fig:expr:main}
\end{figure*}

\paragraph{Competitive methods for SciFact}
Figure~\ref{fig:expr:main} reports the end-to-end latency of \framework plans and baselines on the SciFact dataset. 
All four plans in \framework in conjunction with RankZephyr reduce the end-to-end latency by up to $7.42 \times$ while maintaining competitive accuracy compared to all baselines. 
This is because \framework minimizes LLM calls while maintaining high accuracy with listwise rankers. 
The trend is consistent across all four physical plans. 
The only baseline that is not strictly dominated by \framework is LOTUS with Qwen3-8B. 
It achieves a slightly higher NDCG@10 at the cost of significantly higher latency, likely leveraging Qwen's reasoning capability for higher-quality rankings. 
The \pairwise baseline has a near-competitive accuracy, but still incurs latency about 2x that of \framework's plans due to its high number of LLM calls. 

\paragraph{Non-competitive methods for SciFact}
Our algorithms, combined with general-purpose LLMs (Qwen3-8B and Zephyr-7B), have catastrophically low accuracy due to general-purpose LLMs' subpar performance on listwise ranking. 
Pointwise baselines have similarly near-zero accuracy and noncompetitive latency, demonstrating pointwise rankers' score calibration problem (\sect{sec:intro}). 

\paragraph{Impact of pre-filtering for SciFact}
By applying \tfilter with $S = 1$ before \tselect or \qsort, we can reduce the latency of the model by up to $1.34 \times$ while maintaining accuracy. 
\tfilter with small $S$ is highly effective likely due to the small $K = 10$ compared to $N = 5,183$ which makes false negatives unlikely.

\paragraph{SemBench Q9 sorting task on Movies}
The two sorting algorithms in \framework (\tselect and \qsort) achieve a trade-off between \pairwise and \pointwise baselines in terms of both latency and accuracy. 
This suggests that for sorting tasks, listwise rankers can reduce latency compared to pairwise prompting. 
However, if an hypothetical high-quality pointwise ranker would exist, then listwise solutions' log-linear cost would not beat pointwise solutions' linear cost in latency. 
On the other hand, listwise methods do incur an accuracy penalty for sorting tasks compared to pairwise quicksort. 
This is because there is a slight negative correlation between list size and ranking quality (see Section~\ref{sec:experiments:data:qsort}). 

\begin{figure*}[bth]
    
    \centering
    \begin{subfigure}[b]{0.33\textwidth}
        \centering
        \includegraphics[width=\textwidth]{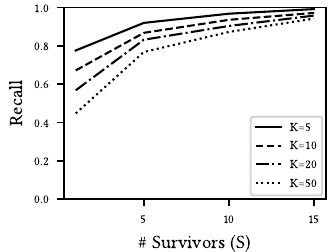}
        \caption{Tournament Filter}
        \label{fig:expr:tfilter}
        \Description{The relationship between recall and the number of survivors S. Increased survivors increases recall up to the asymptote of 1. Larger K decreases recall.}
    \end{subfigure}
    \hfill
    \begin{subfigure}[b]{0.33\textwidth}
        \centering
        \includegraphics[width=\textwidth]{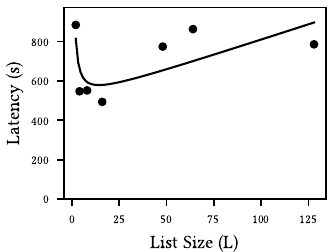}
        \caption{List Size vs. Latency}
        \label{fig:expr:window:latency}
        \Description{Latency vs list size L. There are 6 data points with list size ranging from 2 to 128.}
    \end{subfigure}
    \hfill
        \begin{subfigure}[b]{0.33\textwidth}
        \centering
        \includegraphics[width=\textwidth]{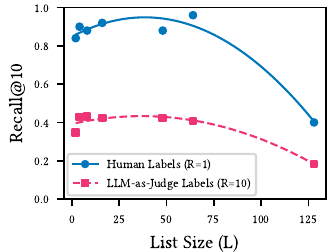}
        \caption{List Size vs. Recall}
        \Description{Recall@10 vs list size L for human and llm-as-judge labels. THere are 7 data points with list size ranging from 2 to 128.}
        \label{fig:expr:window:recall}
    \end{subfigure}
    
    \caption{
    (\ref{fig:expr:tfilter}) The impact of survivor count per bin ($S$) and $K$ on the recall of \tfilter; ground truth labels determined by LLM-as-judge ($1 \le K \le 50$). (\ref{fig:expr:window:latency}, 
    \ref{fig:expr:window:recall}) The imapct of list size ($L$) on the latency and recall of \qselect. 
    \ref{fig:expr:window:latency}'s best fit is a linear model with a reciprocal ($1/L$) term; \ref{fig:expr:window:recall}'s best fit is a quadratic model. 
    All experiments were run on 25 queries of the scifact dataset with $N=5,184$ with $P=P'=1$ against human labels ($R=1$) and LLM-as-judge labels ($R=10$). 
    }
    
    \label{fig:expr:tfilter-window}
\end{figure*}
\begin{figure}
    \centering
    \includegraphics[width=\linewidth]{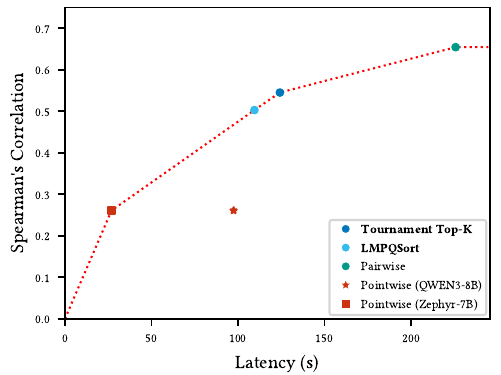}
    \caption{End-to-end comparison of our methods (\qsort, \tselect) with baseline methods on the Movies dataset and Q9 of SemBench. 
    The task is sorting only, $N = 256$ w.r.t. gold standard rankings given in SemBench.}
    \label{fig:expr:sembench}
    \Description{Sparman's correlation vs latency. There are five data points: Tournament Top-K, LMPQSort, pairwise, pointwise with Qwen3-8B, and pointwise with zephyr-7B.}
\end{figure}

\subsubsection{Tournament Filter}
We demonstrated the efficacy of \tfilter for the SciFact task in \fig{fig:expr:main}, yielding a $1.91 \times$ latency reduction with no loss in accuracy. 
\fig{fig:expr:tfilter} shows a broader picture with varying values of $K$ ($K = 1\ldots50$) and survivors per bin ($S = 1\ldots15$). 
As expected, an increase in $S$ improves recall. 
However, as $K$ increases, \tfilter also becomes less efficient at pruning the corpus size. 
Indeed, \fig{fig:expr:main} shows that high $S$ ($S = 10, 15$) increases latency compared to not using the filter at all. 
Notably, the expected recall diminishes most rapidly when $S$ is small and $K$ is large simultaneously. 
As such, for small $K$, \tfilter is extremely effective at pruning documents with nearly no loss. 
On the other hand, it loses efficiency with a higher $K$ as it needs more survivors per bin to reach comparable recall.

\subsubsection{Listwise Multi-Pivot Quickselect}

\begin{figure*}[thb]
    \begin{subfigure}[b]{0.245\textwidth}
        \centering
        \includegraphics[width=\textwidth]{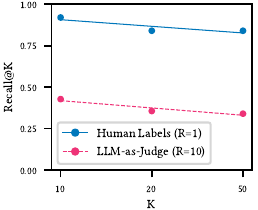}
        \caption{Quickselect w.r.t. $K$}
        \label{fig:expr:qselect:K}
    \end{subfigure}
    \hfill
        \begin{subfigure}[b]{0.245\textwidth}
        \centering
        \includegraphics[width=\textwidth]{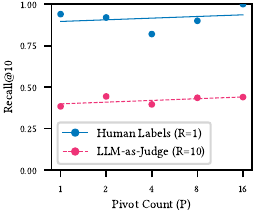}
        \caption{Quickselect w.r.t. $P$}
        \label{fig:expr:qselect:P}
    \end{subfigure}
    \hfill
    \begin{subfigure}[b]{0.245\textwidth}
        \centering
        \includegraphics[width=\textwidth]{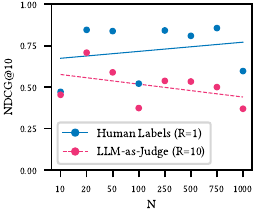}
        \caption{Quicksort w.r.t. $N$}
        \label{fig:expr:qsort:K}
    \end{subfigure}
    \hfill
        \begin{subfigure}[b]{0.245\textwidth}
        \centering
        \includegraphics[width=\textwidth]{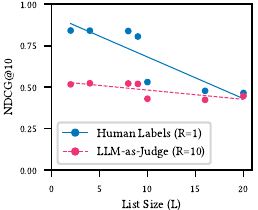}
        \caption{Quicksort w.r.t. $L$}
        \label{fig:expr:qsort:L}
    \end{subfigure}
    
    \caption{Influence of various parameters on the accuracy (Recall@10, NDCG@10) of \qselect and \qsort. 
    Both metrics measured w.r.t. human labels ($R=1$) and LLM-as-judge labels ($R=10$). 
    All experiments use Scifact dataset on 25 test queries with $N=5,184$ and default parameters unless stated otherwise. 
    All best fits use linear models. 
    }

    \Description{From left to right: recall@K vs K for quickselect; recall@10 vs pivot count for quickselect; NDCG@10 vs N for quicksort; NDCG@10 vs list size for quicksort. Empirical data and linear best fit for both human and LLM-as-judge labels.}

    \label{fig:expr:lmpq-params}
\end{figure*}

\paragraph{Impact of $K$, $P$ on \qselect}
Figure~\ref{fig:expr:qselect:K} shows that Recall@K is largely invariant with increasing $K$. 
We believe that constraining the LLM ranker to at most $L$ documents helps maintain recall regardless of $K$. 
Similarly, as shown in Figure~~\ref{fig:expr:qselect:P}, Recall@10 is invariant to the choice of $P$. 
These results suggest that the parameters $K$ and $P$ only affect the latency of \qselect, and their effect on recall is minimal. 

\paragraph{Impact of list size on \qselect}
We also studied the impact of list size ($L$) on the latency and recall of \qselect. 
As shown in \fig{fig:expr:window:latency}, the latency of \qselect is minimized at $L = 20$ with increases for smaller or larger $L$. 
If $L$ is too small, then the information content per token is low, necessitating more LLM calls to complete the algorithm. 
If $L$ is too large, then the quadratic inference complexity of LLMs dominates, dramatically increasing the cost per LLM call. 
The quadratic term cancels out with the fact that increasing $L$ linearly decreases the number of LLM calls (\thm{thm:quickselect-expectation}). 
Hence, the latency in terms of $L$ is modeled as a function of the form $a L + b + c / L$ with an empirical minimum near $L = 20$. 
Next, \fig{fig:expr:window:recall} shows that Recall@10 is largely invariant for $L < 64$ and decreases for $L > 64$. 
Overall, these results suggest that while recall for \qselect can be maintained at higher $L$ values, remaining near the $L \approx 20$ range is preferable to minimize latency. 
As such, \framework uses $L = 20$ by default. 

\subsubsection{Listwise Multi-Pivot Quicksort} \label{sec:experiments:data:qsort}

\paragraph{Impact of $N$, $L$ on \qsort}
We studied the impact of various parameters on NDCG@10 of \qsort on the SciFact dataset. 
Note that we use NDCG instead of Spearman's correlation when $K < N$ as \qsort is primarily intended as a reranking method for the top-K set returned by \qselect. 
As shown in \fig{fig:expr:qsort:K}, NDCG@10 is largely invariant to the size of the document corpus. 
On the other hand, \fig{fig:expr:qsort:L} shows that NDCG@10 diminishes as list size ($L$) increases. 
This data is consistent with \fig{fig:expr:sembench} and \fig{fig:expr:main}, where the \pairwise baseline beats \qsort in terms of ranking quality at the cost of latency. 
This observation motivates our top-few refinement heuristic (\sect{sec:method:optim:finetune}), where we re-rank the top-few (default: top-10) using $L = 2$. 
\fig{fig:expr:heur:few} shows the impact of the refinement heuristic. 
It maintains or dramatically improves the final NDCG at an up to 4\% (constant-bound) latency overhead. 

\subsubsection{Practical optimizations}  \label{sec:expr:result:optim}
\begin{figure*}[thb]
    \begin{subfigure}[b]{0.33\textwidth}
        \includegraphics[width=\textwidth]{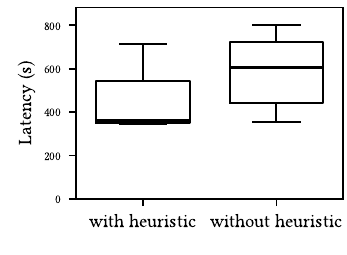}
        \caption{Embedding Heuristic}
        \label{fig:expr:heur:emb}
    \end{subfigure}
    \hfill
    \begin{subfigure}[b]{0.33\textwidth}
        \includegraphics[width=\textwidth]{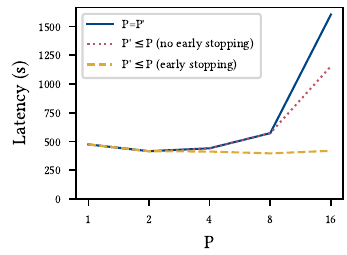}
        \caption{$P' \neq P$ + Early Stopping}
        \label{fig:expr:heur:early}
    \end{subfigure}
    \hfill
    \begin{subfigure}[b]{0.33\textwidth}
        \includegraphics[width=\textwidth]{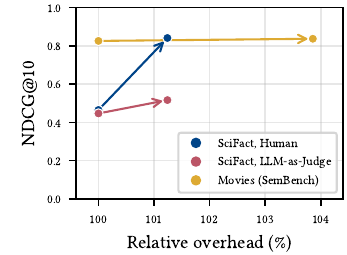}
        \caption{Top-few Refinement Heuristic}
        \label{fig:expr:heur:few}
    \end{subfigure}
    
    \caption{
    The impact of various implementation tricks on \qselect and \qsort. (\ref{fig:expr:heur:emb}) 
    %\todo{I do not understand what the bars mean.} 
    Latency reduction from using embedding-distance based pivot selection for \qselect. 
    The box plot displays the range and quantiles across the 25 Scifact queries used. 
    (\ref{fig:expr:heur:early}) Latency reduction from using physical pivot counts ($P'$) smaller than logical pivot counts ($P$) in conjunction with early stopping for \qselect. 
    (\ref{fig:expr:heur:few}) 
    Impact of top-few reranking with pairwise quicksort on the NDCG of final ranking for \qsort with the arrows show the change in the value of NDCG. %\todo{explain relative overhead} \fn{The arrows show the change in the value ofNDCG.}
    }
    \label{fig:expr:heur}

    \Description{From left to right: Box plot of latency distribution with and without embedding heuristic. Latency vs pivot count P for the P=P' regime, P' <= P regime with and without early stopping. NDCG@10 vs relative overhead in \% for scifact with human and llm-as-judge labels as well as movies dataset on sembench.}
\end{figure*}

Figure~\ref{fig:expr:heur} shows the impact of practical optimizations we applied to \framework (\sect{sec:method:optim}) on the SciFact dataset. 
Figure~\ref{fig:expr:heur:emb} shows that utilizing an embedding heuristic reduces latency: $1.28 \times$ on average and $1.1 \times$ for the worst-case tail latency. 
This implies that the LLM ordering of documents within the dataset is correlated with the ordering of the documents based on embedding similarities. 
Figure~\ref{fig:expr:heur:early} shows that using physical pivot count $P'$ smaller than the logical pivot count $P$ reduces latency for extremely large values of $P \ge 16$ for \qselect. 
Early stopping makes larger pivot counts ($P \ge 8$) slightly better than the theoretically optimal pivot count of 4. 
This makes sense: both $P' < P$ and early stopping mitigate the downside of having many pivots, which is that only a few non-pivot documents can fit into each LLM call. 
Figure~\ref{fig:expr:heur:few} shows the impact of the top-few pairwise refinement heuristic, discussed in the previous section. 

\section{Discussion} \label{sec:discussion}

We now discuss several limitations and possible improvements to our work. 

\paragraph{Benchmarks}
The existing benchmark datasets, such as \emph{SciFact} and \emph{Movies}, are possibly too easy. 
The former has been exhaustively studied by the IR community~\cite{thakur2021beir}. 
The latter uses sentiments to rank movie reviews; sentiment analysis is also fairly well-studied. 
While our experiments demonstrate the efficacy of our methods, a more difficult benchmark that limit tests the capabilities of LLMs -- or any other method -- may be more informative. 
Furthermore, the default method for ranking tasks in SemBench is via explicit pointwise scoring, which we showed in \sect{sec:expr} is a weak baseline. 
SemOp benchmarks could be extended to include a greater variety ORDER BY and Top-K queries. 

\paragraph{Better listwise LLM rankers}
While RankZephyr remains the go-to open-weight listwise ranker, there has been active research into improving accuracy and latency. 
Some promising approaches include a more sophisticated fine-tuning loss~\cite{alro} and the FIRST method that generates a single token for a listwise sort~\cite{first}. 
We do not use them in this work as the former is currently not open-sourced, and the latter incurs some accuracy penalties. 
\fixed{This is because FIRST approximates a standard listwise ranker using the first token's logits as a proxy.}
Emerging prompting paradigms~\cite{setwise} may also be considered for future work. 

\paragraph{Improved self-contradiction resolution}
We bluntly enforce self-consistency in the event of contradictory rankings (\sect{sec:method:optim}). 
A more sophisticated method would aggregate rankings in ways that can handle self-contradictions. 
Works by \citet{liu2024aligning} and \citet{tang2024found} are explorations in this direction, but require significantly more LLM inference. 
While our methods assume self-consistency for the sake of reducing latency, \fn{a more effective compromise remains open for exploration.} 

\paragraph{Analysis of noisy rankers}
Our analysis assumes that the listwise ranker is always correct. 
This is not true for LLMs, as even fine-tuned listwise rankers can output rankings that are ``incorrect'' relative to the gold-standard. 
Existing methods for reducing noise via repeated prompting are computationally infeasible~\cite{tang2024found}. 
A more sophisticated analysis could calculate the accuracy of our methods for imperfect rankers.

\paragraph{Improving analysis of quickselect}
\fixed{In our analysis of quickselect, we make the key simplifying assumption of self-similarity.
While this assumption roughly matches reality (\fig{fig:simu:qselect}), it is not the most principled approach.}
$\psi_i$ converges to $1/2$ over a sufficiently large number of rounds as the influence of the initialization shrinks in comparison to randomness. 
To model this, the theory of Markov chains is needed. 

\section{Conclusion} \label{sec:conclusion}

We study how to optimize the execution of semantic top-K operators using listwise prompting. 
To solve the issue of poor listwise ranking performance by general-purpose LLMs, we leverage recent fine-tuned listwise rankers. 
We study a collection of listwise rank aggregation algorithms that leverage the listwise prompting interface to reduce the number of LLM calls. 
We provide analysis on the cost and recall of the algorithms. 
We leverage the analysis to build our \framework framework which minimizes cost while meeting a recall target by solving an optimization problem. 
Through extensive experiments, we demonstrate the superiority of our algorithms to a variety of baseline methods in standard benchmarks.

%\begin{acks}
% This work was supported by the [...] Research Fund of [...] (Number [...]). Additional funding was provided by [...] and [...]. We also thank [...] for contributing [...].
%\end{acks}

\bibliographystyle{ACM-Reference-Format}
\bibliography{refs}

\clearpage

\section{Appendix} \label{sec:appendix}

\subsection{Proof of Theorem~\ref{thm:quicksort-expectation}} \label{sec:appendix:proof1}
For notational simplicity, we use $T_N$ to denote the full cost function $\text{Cost}_{\text{LMPQSort}} (N, L, P)$, omitting the fixed variables $L$ and $P$. 

In each round of \qsort, the problem is divided into $P+1$ subproblems.
Let $g_i$ denote the size of the $i$th subproblem.
Since $g_i$ is symmetric w.r.t. $i$, we denote $g_i$ with arbitrary $i$ as $g$.
We also assume that the input list is randomly shuffled, which has been shown to imply that the subproblems are also randomly shuffled~\cite{hennequin1989combinatorial}.
Also note that we always simplify all non-dominating terms (w.r.t. $N$) in big-O as our goal is to uncover the leading coefficient.

\begin{proposition}
\begin{equation}
\mathbb{E}[g] = \frac{N-P}{P+1}.
\end{equation}
\label{prop:E[g]}
\end{proposition}

\begin{proof}
There are $N - P$ non-pivot elements distributed among $P + 1$ subproblems.
By symmetry of $g_i$ w.r.t. $i$, each subproblem has equal expected size.
\end{proof}

\begin{proposition}
\begin{equation}
\mathbb{E}[g \log g] = \frac{N}{P+1} \log \left(\frac{N}{P+1}\right) + O(N).
\end{equation}
\label{prop:E[glogg]}
\end{proposition}

\begin{proof}
We approximate $\mathbb{E}[g \log g]$ via a second-order Taylor expansion of $f(g) = g \log g$ near $g \approx \mu = \mathbb{E}[g]$. Then:
\begin{equation}
\mathbb{E}[f(g)] = f(\mu) + f’(\mu)\mathbb{E}[g - \mu] + \frac{1}{2}f’'(\mu)\mathbb{E}[(g - \mu)^2]
\end{equation}

The first term can be directly substituted as $f(\mu) = \mu \log \mu$; the second term cancels out since $\mathbb{E}[g - \mu] = 0$; the third term simplifies by $f''(g) = 1/g$ and $\mathbb{E}[(g - \mu)^2] = \mathrm{Var}(g)$. Hence, 
\begin{equation}
    \mathbb{E}[g \log g] = \mu \log \mu + \frac{1}{2\mu} \cdot \mathrm{Var}(g).
    \label{eq:thm1:glogg}
\end{equation}

We now compute $\mathrm{Var}(g)$. By symmetry, $\mathrm{Var}(g) = \mathrm{Var}(g_0)$. Let $Y_i = \mathbb{I}[\text{element } i \text{ falls in gap } 0]$, so that $g_0 = \sum_{i=1}^{N} Y_i$. Then
\begin{equation}
    \mathrm{Var}(g_0) = \sum_i \mathrm{Var}(Y_i) + 2\sum_{i < j} \mathrm{Cov}(Y_i, Y_j).
    \label{eq:thm1:var}
\end{equation}

By the preservation of randomness for subproblems,
\begin{equation}
    \mathbb{E}[Y_i] = P[\text{any element falls in gap }0] = \frac{N-P}{N(P+1)}.
\end{equation}

This allows us to solve for the variance term:
\begin{align}
    \mathrm{Var}(Y_i) &= \sum_{i=1}^N \mathbb{E}[Y_i] (1-\mathbb{E}[Y_i])\\
    &= N \cdot \frac{N-P}{N(P+1)} \cdot \left(1 - \frac{N-P}{N(P+1)} \right) \\
    &= \frac{N^2+P}{N(P+1)} \\
    &= O(N). \label{eq:thm1:var:simplified}
\end{align}

Next, we solve for the covariance term. 
\begin{equation}
    2 \sum_{i < j} \mathrm{Cov}(Y_i, Y_j) = 2 \sum_{i < j} \left[ \mathbb{E}[Y_i Y_j] - \mathbb{E}[Y_i]\mathbb{E}[Y_j] \right]
\end{equation}

Since the preservation of randomness for subproblems implies symmetry w.r.t. pairs $(i, j)$, 
\begin{equation}
    2 \sum_{i < j} \mathrm{Cov}(Y_i, Y_j) = 2 N (N-1) \left[ \mathbb{E}[Y_i Y_j] - \mathbb{E}[Y_i]\mathbb{E}[Y_j] \right].
    \label{eq:thm1:cov}
\end{equation}

Now assume a continuous relaxation as $N \to \infty$, and where the $i$\thh and $j$\thh elements, as well as the $P$ pivots, are independent uniform samples from $[0, 1]$. 

Now, wlog, let $x < y$ be the values of the $i$\thh and the $j$\thh elements, respectively. 
Then $\mathbb{E}[Y_i] = (1-x)^P$ since it is the probability that all $P$ pivots are above $x$. Similarly, $\mathbb{E}[Y_j] = (1-y)^P$. 
$\mathbb{E}[Y_i Y_j] = (1-y)^P$ since it is the probability that all $P$ pivots are above whichever is larger between $x$ and $y$. 

Now we can obtain an integral approximation of the joint expectation term in Equation~\ref{eq:thm1:cov}:
\begin{equation}
    \left[ \mathbb{E}[Y_i Y_j] - \mathbb{E}[Y_i]\mathbb{E}[Y_j] \right] = \int_0^1 \int_0^y (1-y)^P - (1-x)^P (1-y)^P dx dy.
\end{equation}

Hence, by substitution, an integral form for Equation~\ref{eq:thm1:cov}:
\begin{equation}
    2 \sum_{i < j} \mathrm{Cov}(Y_i, Y_j) = 2 N(N-1) \int_0^1 \int_0^y (1-y)^P - (1-x)^P (1-y)^P dx dy
\end{equation}

Since the integral has no dependence on $N$ (only $P$), w.r.t. N, Equation~\ref{eq:thm1:cov} simplifies as: 
\begin{equation}
    2 \sum_{i < j} \mathrm{Cov}(Y_i, Y_j) = O(N^2).
    \label{eq:thm1:cov-simplified}
\end{equation}

Hence by substitution of Equation~\ref{eq:thm1:cov-simplified} and \ref{eq:thm1:var:simplified} into Equation~\ref{eq:thm1:var}:
\begin{align}
    \mathrm{Var}(g_0) &= O(N) + O(N^2) = O(N^2).
\end{align}

Now the second-order Taylor term in Equation~\ref{eq:thm1:glogg} also simplifies:
\begin{align}
    \frac{1}{2\mu} \cdot \mathrm{Var}(g) &= \frac{2(P+1)}{N-P} \cdot O(N^2) = O(N). 
\end{align}

Therefore, we have fully derived Equation~\ref{eq:thm1:glogg} into the desired form:
\begin{align}
    \mathbb{E}[g \log g] &= \frac{N-P}{P+1} \log\left(\frac{N-P}{P+1}\right) + O(N)\\
    &= \frac{N}{P+1} \log\left(\frac{N}{P+1}\right) + O(N)
\end{align}
where $N - P \approx N$ when $P \ll N$. 

\end{proof}

With Propositions~\ref{prop:E[g]} and~\ref{prop:E[glogg]} established, we proceed with the main proof.

\begin{proof}[Proof of Theorem~\ref{thm:quicksort-expectation}]
The general recurrence relation for $T_N$ is
\begin{equation}
T_N = \frac{N}{L - P} + (P+1) \cdot \mathbb{E}[T_g]
\label{eq:thm1:recurrence}
\end{equation}
The first term arises because bucket-sorting $N$ elements using windows of size $L$ with $P$ pivots requires $\lceil N/(L-P) \rceil$ oracle calls, approximating away the ceiling for large $N$.

Assume the ansatz
\begin{equation}
    T_N = \alpha N \log N + \beta N + o(N)
    \label{eq:thm1:ansatz}
\end{equation}
where $\alpha$ and $\beta$ are constants to be determined.

Now consider the subproblem for some gap $g$. 
\begin{align}
    \mathbb{E}[T_g] &= \alpha \, \mathbb{E}[g \log g] + \beta \, \mathbb{E}[g] + o(N) \\
    &= \alpha \cdot \frac{N}{P+1} \log\left(\frac{N}{P+1}\right) + \beta \cdot \frac{N-P}{P+1} + o(N)
\end{align}
where we used Propositions~\ref{prop:E[g]} and~\ref{prop:E[glogg]}.

Expanding the logarithm yields
\begin{equation}
    \mathbb{E}[T_g] = \frac{\alpha N}{P+1} \log N - \frac{\alpha N}{P+1} \log(P+1) + \frac{\beta(N-P)}{P+1} + o(N).
    \label{eq:thm1:ETg}
\end{equation}

Substituting Equation~\ref{eq:thm1:ETg} back into Equation~\ref{eq:thm1:recurrence}:
\begin{align}
    T_N &= \frac{N}{L-P} + (P+1) \Bigg[ \frac{\alpha N}{P+1} \log N - \frac{\alpha N}{P+1} \log(P+1) \nonumber\\
    &\qquad\qquad\qquad\qquad + \frac{\beta(N-P)}{P+1} + o(N) \Bigg] \\
    &= \frac{N}{L-P} + \alpha N \log N - \alpha N \log(P+1) + \beta(N-P) + o(N) \\
    &= \alpha N \log N + \left[ \frac{1}{L-P} - \alpha \log(P+1) + \beta \right] N - \beta P + o(N). \label{eq:thm1:TN}
\end{align}

For the ansatz $T_N = \alpha N \log N + \beta N + o(N)$ to be self-consistent, we match coefficients. The coefficient of $N \log N$ is $\alpha$ on both sides, which is satisfied. Matching the coefficient of $N$ requires, from Equation~\ref{eq:thm1:TN} and Equation~\ref{eq:thm1:ansatz}:
\begin{equation}
    \beta = \frac{1}{L-P} - \alpha \log(P+1) + \beta,
\end{equation}
which simplifies to
\begin{equation}
    \frac{1}{L-P} - \alpha \log(P+1) = 0.
\end{equation}

Solving for $\alpha$:
\begin{equation}
    \alpha = \frac{1}{(L-P)\log(P+1)}.
\end{equation}

The remaining $-\beta P$ term is $O(1)$ w.r.t. $N$ and is absorbed into $o(N)$, so the ansatz is consistent.

Therefore,
\begin{equation}
    T_N = \frac{1}{(L-P)\log(P+1)} \cdot N \log N + \beta N + o(N). 
\end{equation}
$\beta N + o(N) = O(N)$ which completes the proof.

\end{proof}

Here, $\beta$ is not determined by the general recurrence relation, but instead by the boundary condition.
It is more arduous to derive at relatively minor gains in the efficacy of the model.
For simplicity, we simply fit a universal $\beta$ with empirical data in [cite].

\subsection{Proof of Theorem~\ref{thm:quickselect-expectation}} \label{sec:appendix:proof2}

For notational simplicity, we use $T_{N, K}$ to denote the full cost function $\mathrm{Cost}_{\text{LMPQSelect}} (N, K, L, P)$, omitting the fixed variables $L$ and $P$.
Let $g\psi$ denote the size of the bucket containing $\psi = K/N$.
Now let $\rho_\psi = \mathbb{E}[g_\psi]/N$.

\begin{proposition}
\begin{equation}
\mathbb{E}[g_\psi] = N \cdot \frac{2 - \psi^{P+1} - (1-\psi)^{P+1}}{P+1}.
\end{equation}
\label{prop:containing-bucket}
\end{proposition}

\begin{proof}
Assume a continuous relaxation of the problem with $P$ pivots drawn uniformly at random from $[0,1]$. Then, the expected length of the interval containing $\psi$ is as follows.
\begin{align}
\rho_\psi &= \int_0^1 \mathbb{I}[\text{no pivot between } \psi \text{ and } x] dx\\
\mathbb{E}[|\rho_\psi|] &= \int_0^1 \mathbb{E}[\mathbb{I}[\text{no pivot between } \psi \text{ and } x] dx\\
&= \int_0^1 P[\text{no pivot between } \psi \text{ and } x] dx \label{eq:thm2-int}
\end{align}

For $x < \psi$, all $P$ pivots must avoid $(x, \psi)$, which occurs with probability $(1 - (\psi-x))^P$.
For $x > \psi$, all $P$ pivots must avoid $(\psi, x)$, which occurs with probability $(1 - (x-\psi))^P$.

Hence Equation~\ref{eq:thm2-int} becomes
\begin{equation}
    \mathbb{E}[g_\psi] = \int_0^\psi (1 - \psi + x)^P\, dx + \int_\psi^1 (1 - x + \psi)^P\, dx.
    \label{eq:thm2-int0}
\end{equation}

Compute the integrals:
\begin{align}
    \int_0^\psi (1 - \psi + x)^P\, dx &= \frac{1 - (1-\psi)^{P+1}}{P+1}, \label{eq:thm2-int1}\\
    \int_\psi^1 (1 - x + \psi)^P\, dx &= \frac{1 - \psi^{P+1}}{P+1}. \label{eq:thm2-int2}
\end{align}

Substituting Equations~\ref{eq:thm2-int1} and \ref{eq:thm2-int2} to Equation~\ref{eq:thm2-int0}:
\begin{equation}
    \mathbb{E}[L_\psi] = \frac{2 - \psi^{P+1} - (1-\psi)^{P+1}}{P+1}.
\end{equation}

Since bucket size is $g_\psi = N \cdot \rho_k$, we have $\mathbb{E}[g_\psi] = N \cdot \mathbb{E}[L(k)]$ which yields the desired result. 

\end{proof}

With Proposition~\ref{prop:containing-bucket} established, we proceed with the main proof.

\begin{proof}[Proof of Theorem~\ref{thm:quickselect-expectation}]
The recurrence relation for $T_{K,N}$ is
\begin{equation}
T_{N,K} = \frac{N}{L-P} + \mathbb{E}[T_{g_\psi, K’}].
\label{eq:thm2-recurrence}
\end{equation}
Here, $g_\psi$ is the size of the bucket containing position $K$, and $K’$ is the adjusted target within that bucket. The first term arises because bucket-sorting $N$ elements requires $1 + (N-P)/(L-P) \approx N/(L-P)$ oracle calls for sufficiently large $N$.

Assume the ansatz
\begin{equation}
    T_{K,N} = \alpha_\psi \cdot N + O(1).
    \label{eq:thm2-ansatz}
\end{equation}

Now substitute the subproblem containing position $K$ into Equation~\ref{eq:thm2-ansatz}. 
Under the self-similarity assumption, the recursive subproblem has the same cost structure, so
\begin{equation}
    \mathbb{E}[T_{g_\psi,K'}] = \alpha_\psi \cdot \mathbb{E}[g_\psi] + O(1).
    \label{eq:thm2-ansatz-subproblem}
\end{equation}

Substituting Equation~\ref{eq:thm2-ansatz-subproblem} into Equation~\ref{eq:thm2-recurrence}:
\begin{equation}
    \alpha_\psi \cdot N = \frac{N}{L-P} + \alpha_\psi \cdot \mathbb{E}[g_\psi] + O(1).
    \label{eq:thm2-1}
\end{equation}

Using the definition of $\rho_\psi = \mathbb{E}[g(k)]/N$ and solving for $\alpha_\psi$ from Equation~\ref{eq:thm2-1} yields
\begin{equation}
    \alpha_\psi = \frac{1}{(L-P)(1 - \rho_\psi)} + O(1/N).
    \label{eq:thm2:alpha_psi}
\end{equation} 

By Proposition~\ref{prop:containing-bucket},
\begin{equation}
    \rho_\psi = \frac{2 - \psi^{P+1} - (1-\psi)^{P+1}}{P+1}.
\end{equation}

Thus,
\begin{align}
    1 - \rho_\psi &= 1 - \frac{2 - \psi^{P+1} - (1-\psi)^{P+1}}{P+1} \nonumber \\
    &= \frac{P + 1 - 2 + \psi^{P+1} + (1-\psi)^{P+1}}{P+1} \nonumber\\
    &= \frac{P - 1 + \psi^{P+1} + (1-\psi)^{P+1}}{P+1}.
    \label{eq:thm2:1-rho_psi}
\end{align}

Substituting Equation~\ref{eq:thm2:1-rho_psi} to Equation~\ref{eq:thm2:alpha_psi} results in the desired formula for $\alpha_\psi$. 
\begin{equation}
    \alpha_\psi = \frac{P+1}{(L-P)(P - 1 + k^{P+1} + (1-k)^{P+1})}.
    \label{eq:thm2:alpha_psi_2}
\end{equation}

Substituting Equation~\ref{eq:thm2:alpha_psi_2} to Equation~\ref{eq:thm2-ansatz} yields
\begin{equation}
    T_{N, K} = \frac{P+1}{(L-P)(P - 1 + \psi^{P+1} + (1-\psi)^{P+1})} \cdot N + O(1)
\end{equation}
which matches the ansatz and concludes the proof.

%We verify consistency with the boundary condition. For $N \leq L$, we have $T_{K,N} = 1$, while $cN = O(L) = O(1)$ for constant $L$. Thus, the $O(1)$ remainder absorbs the boundary, confirming consistency.

\end{proof}

\subsection{Full Listwise Ranking Prompt Template}
We utilize RankLLM’s default listwise ranking prompt template:
\begin{verbatim}
method: “singleturn_listwise”
system_message: “You are RankLLM, an intelligent
assistant that can rank passages based on their
relevancy to the query”
prefix: |
I will provide you with {num} passages, each indicated
by a numerical identifier []. Rank the passages based
on their relevance to the search query: {query}.
body: |
[{rank}] {candidate}
suffix: |-
Search Query: {query}.
Rank the {num} passages above based on their relevance
to the search query. All the passages should be
included and listed using identifiers, in descending
order of relevance. The output format should be
[] > [], e.g., [2] > [1], Answer concisely and
directly and only respond with the ranking results,
do not say any word or explain.
output_validation_regex: r"^[\d+]( > [\d+])*$"
output_extraction_regex: r"[(\d+)]"
\end{verbatim}

\subsection{LLM-as-Judge Augmentation Details} \label{sec:appendix:augmentation}

We used four retrievers to retrieve the top-1k most relevant documents each from the corpus of 5k documents: BM25, SPLADE, E5, and BGE. The union of the retrieved documents consisted of about 2k unique documents. We then aggregated the retrievers’ rankings using reciprocal rank fusion (RRF) with the standard damping factor of $k = 60$.

We took the top-1k documents from the retrieval step and reranked them using four rerankers: MiniLM-6, MiniLM-12, BGE, and MMarco-MiniLM.
We again aggregated their (re-)rankings using RRF.

We took the top-500 documents from the reranking step and ranked them using pairwise quicksort with GPT-4o.

Finally, we took the top-50 documents from the quicksort step and compared them exhaustively using pairwise prompts on GPT-4o, yielding 2500 cached pairwise comparisons.
They were then aggregated into a top-50 ranking using Elo ratings and random matches until convergence.

\end{document}